\documentclass[10pt]{amsart}

\usepackage[english]{babel}
\usepackage[T1]{fontenc}
\usepackage{lmodern,textcomp}
\usepackage{amsmath}
\usepackage{amsthm}
\usepackage{amssymb}
\usepackage{ltugcomn}
\usepackage{stmaryrd}
\usepackage{dsfont}
\usepackage{enumitem}
\usepackage{graphicx}
\usepackage{caption}
\usepackage{subcaption}
\usepackage{titletoc}
\usepackage{listings}
\usepackage[super]{nth}
\graphicspath{{fig/}}
\usepackage{tikz}
\usetikzlibrary{decorations.markings}
\usepackage{pgfplots}
\usetikzlibrary{arrows,calc}
\usepackage[latin1]{inputenc}
\usepackage[T1]{fontenc}
\usepackage{verbatim,url}
\usepackage{amsfonts,amsmath,epsfig,bbm,graphics}

\usepackage[pdfpagelabels,plainpages=false,colorlinks=true,%
  colorlinks=true,%
  linkcolor=black!70!red!100,%
  urlcolor=black!70!blue!100,%
  citecolor=black!70!green!100]{hyperref}

\newif\ifRapportRecherche
\RapportRecherchefalse 
\RapportRecherchetrue  

\ifRapportRecherche

\else

\fi

\def\cal{\mathcal}

 \newtheorem{proposition}{Proposition}
 \newtheorem{lemma}{Lemma}
\newtheorem{definition}{Definition}
 \newtheorem{corollary}{Corollary}

\title{Two-choice regulation in heterogeneous closed networks}
\author{Christine Fricker  \and Nicolas Servel}
\address{C. Fricker and Nicolas Servel are with INRIA Paris-Rocquencourt Domaine de
  Voluceau, 78153 Le Chesnay, France.
  \url{Christine.Fricker@inria.fr}}%

\date{\today}

\begin{document}

\begin{abstract}
A heterogeneous closed  network with $N$ one-server queues with finite capacity and one infinite-server queue is studied. A target application is bike-sharing systems. Heterogeneity is taken into account through clusters whose queues have the same parameters. Incentives to the customer to go to the least loaded one-server queue among two chosen within a cluster are investigated. By mean-field  arguments,  the  limiting queue length    stationary distribution as $N$ gets large is analytically tractable. Moreover,  when all customers follow incentives, the probability that a queue is empty or full is approximated. Sizing the system to  improve performance is reachable under this policy. 
\end{abstract}

\maketitle
\section{Introduction}
\label{sec:intro}
{\bf Product-form networks.}  A large literature deals with closed Jackson networks also called Gordon-Newell networks. First they present an explicit  product-form stationary distribution (see \cite{Kelly-1}). Second they appear as complex systems in a wide range of applications as communication, computer, biology  and transport networks.

The problem is that this nice property turns out to be useless for large systems because of the exponential growth of the state space. Some  recent research describes the asymptotic behavior as these networks get large, in general both the number of nodes and customers and the ratio of the two numbers tending to some constant $\lambda$ (see \cite{Kogan-1}, \cite{Malyshev-6}, \cite{Fayolle-7} and reference therein). The main result is the existence of a critical value of $\lambda$ such that under this value, the system is stable with each finite set of queues being asymptotically independent with  geometric distributions for queue lengths, while above this value, some queues, those with the maximum so-called utilization, behave as bottlenecks with a infinite mean number of customers while the others are stable. See also \cite {Anselmi-1} for recent advances on a multi-class model. 

{\bf Mean-field techniques.} However, such product-form networks do not cover the large framework of applications. Without such a invariant measure, for homogeneous models, mean-field techniques allow to obtain the limiting steady-state of the system as it gets large. They are popular for two decades. The main idea is to obtain the limit of the empirical measure process as a dynamical system,  deterministic solution of an ODE, then to study the equilibrium points of this system. In general, it is proved that there is a  unique equilibrium point. Less easy  to prove is that the invariant measure of the empirical measure process  concentrates as the network gets large  to this equilibrium point. It remains in most cases the hard part of the work. Such convergence is out of the scope of the 
paper here. 

When the network is not homogeneous, the techniques can extend but are sometimes difficult to apply. To preserve a discrete state space, the network is divided into clusters whose number remains fixed as the system gets large (see \cite{Gast-1} in another context).

{\bf Motivation.} The target application here is bike-sharing systems. They have seen explosive growth last years, giving place to lot of research. Velib' in Paris plays an important role, launching the largest program, even now the largestone  outside China. The problem is to manage these systems  in order to maintain resources, both bikes and available spots, where the users need them. Redistribution by trucks is developed, though it affects only  a small number of bikes (around 3 000 moves by trucks per day in Velib compared to 100 000 bike trips). Incentives to users  by the operator are more promising because more scalable. It is called natural regulation,  made possible because the user can know the system state in real time via  his smartphone for example. The issue which is addressed here is the impact of a two-choice strategy. It  consists for the user to choose two stations to return his bike, then returning it in the least loaded one. If it regulates very well a homogeneous system (see  \cite{FrickerGast-velib} for details), one can wonder what happens in a heterogeneous system. 

More precisely, the system consists in the following: A user goes to a station, takes a bike and returns it to another station. But he has to face the lack of resources, both bikes and available spots. Indeed, when arriving in a station with no bikes, the user has to decide, either to leave the system or to find a bike in a station nearby. Then, when returning in a station with no available spot, to return in a station nearby. Note that he has to return his bike due to  the penalty in case he leaves the bike. The point  here is that the user chooses his return station, by choosing between two stations near its destination the least loaded one. The operator can help the  users to behave like by incentives. Therefore the case where only a fraction  of the customers choose is investigated.

 Furthermore the dimensioning problem for the operator is addressed: to find the  system fleet size and  the station capacities in order to obtain a system which performs well.

{\bf A model.}
Such a system can be seen as a closed network. Customers are bikes, with a  fixed fleet size.  They go to  two kinds of nodes: stations  as one-server queues where the service time is the user inter arrival time  to this station, and routes as infinite-server queues where the service time is the trip time on this route. The model is proposed in  \cite{George-1} for a vehicle-sharing system. More precisely, in  \cite{George-1},   the user leaves the system if no bike is available when he arrives at a station, and with stations of infinite capacity, a spot is always available when returning a bike. 

The main drawback of the previous models is that it does not take into account the finite capacity of the stations and the related strategies of the users to be able to return their bikes. 
For that, the model proposed here is more simple in some sense, while the biking users are not distinguished, but it is still  a closed network, with  one-server queues for stations, and one infinite-server queue containing the pool of biking users. 

In this model, when joining a saturated one-server queue, the customer reattempts in another queue, after a time with  the same distribution, until he returns his bike. The model could be refined, taking into account the research of a bike for a customer  arriving at a station with no bike, or taking a  mean time for looking for a station to return the bike smaller than the mean trip time. We claim that these refinements do  not change  basically the study.

The network is heterogeneous. Indeed, queues are grouped in   clusters. In each cluster, the queues have the same parameters. Basically,  clusters are large and the customer chooses two stations within a cluster. Other models, may be more close for realistic ones, where some stations are isolated among stations with other parameters are out of the scope of the paper. Even if mean-field techniques apply, the limiting steady-state would be analytically untractable except by  numerical simulations of the dynamical system.

{\bf The results.}
The aim of the paper is to investigate the impact of a two choice strategy on the limiting steady-state of the queue lengths as the network gets large. 

The proposed method is mean-field. The limit as $N$ gets large of the empirical measure process is obtained and it is proved that it has a unique equilibrium point. The key argument is that the equilibrium point has a probabilistic interpretation, in term of the equilibrium point of the dynamical system associated to a open set of $N$ queues when joining the least loaded queue among two queues, with a arrival-to-service ratio $\rho r_i$ in cluster $i$, where $\rho$ is the unique solution of a fixed point equation. For that, a monotonicity property in the arrival rate is used. Then the equilibrium point is indeed the concentration point of the queue length distributions as the system gets large. These results can be straightforwardly extended to the case where only a portion $\beta_i$ of customers joining cluster $i$ choose among two queues while the others go to one single queue.

Once these convergence results are stated, an interesting part is the analysis of the limit. The limiting queue length distribution can be viewed as a function of the total proportion of customers per queue $s$ by a parametric curve with parameter $\rho$. It allows qualitative and quantitative results on the system limiting behavior. We analyze the probability that a station is empty or full (of customers) called {\em system performance}. This quantity is minimal in a cluster for a very short interval of $\rho$, corresponding nevertheless to a wide interval of $s$. The value of the probabilities for queues in different clusters can be approximately derived as a function of $s$ on these plateaux, so this function is quite well understanding. 

The influence of the different parameters, as the number of queues per cluster, the capacities of the queues and the  proportion of customers per queue $s$, can be discussed. The main conclusion is that choosing the capacities and the total fleet size can just allow to obtain the best performance in one cluster. To decrease this value, the number of stations per cluster must be changed, in order to balance the system.
  
{\bf Related works.} Though optimization is the main hot topic about bike-sharing systems (see \cite {Nair-1}, \cite{Raviv-1},  \cite{Raviv-2}, \cite{Chemla-1}, \cite{Jost-1} and others),  few has been done to understand it as a heterogeneous stochastic network. It is mainly due to the complexity of the system.

It is  to our knowledge the first stochastic model  devoted to bike-sharing systems is in \cite{George-1}. As a BCMP network \cite{Baskett-1}, such a network has a product-form  stationary measure. For a study of this large system, see \cite{Malyshev-6} and reference therein.

There are two companion papers of this one. First \cite{FrickerGast-velib} deals with a homogeneous model where different regulation strategies are explored.  Simple models to  study the system behavior under incentives to choose and redistribution by trucks  are proposed and analyzed. Second, a model with clusters is proposed in \cite{velib-aofa} but for a basic bike-sharing system, not taken into account incentives or redistribution. Paper \cite{velib-aofa} is a first step to work with a heterogeneous model. It allows to understand how to manage heterogeneity. Our present work has a completely different content, focusing on two-choice strategy, even if Section~\ref{sec:convergence}   has some similarities on methods used in \cite{FrickerGast-velib}, which are clearly presented.

{\bf Outline of the paper.} Section \ref{sec:model} deals with the model description. Section~\ref{sec:convergence} presents mean-field convergence results. In Section~\ref{sec:performance}, the limiting stationary  queue length distribution is investigated. As an example, for the two-cluster case, the probability that a queue is empty or full  is studied as a function of $s$. It gives its qualitative behavior as well as  quantitative bounds. In Section~\ref{sec:sizing}, the influence of the parameters is discussed, especially for the sizing problem.  Section~\ref{sec:conclusion} gives a discussion, mainly on model limitations, and the conclusion of  the paper.

\section{Model Description}
\label{sec:model}
 This section deals with the description of a closed queueing network. The heterogeneity is modeled by clusters. The key point is to take into account the  finite capacities at the one-server queues and the resulting route of the customers. Moreover  the behavior of the system under a two-choice  strategy when  joining the one-server queues is addressed. 

Consider $N$ one-server queues numbered from $1$ to $N$, one infinite-server queue numbered $0$ and a set of $M$ customers. Unless  queue $0$ is specified, let us call a queue a one-server  queue. In the system, both $N$ and $M$ are large, with the total proportion of customers per queue $M/N$ tending to a constant $s$ as $N$ tends to infinity, which is a key quantity in terms of sizing. The queues are grouped  in $C$ clusters, $C\geq 1$, such that, in each cluster, the queues have the same parameters. There are $N_i$ queues with capacity $K_i$ in cluster $i$, with $N_i/N$ tending to $\alpha_i$, when $N$ gets large. A customer leaves  a queue of cluster $i$ according to a Poisson process with rate $\lambda_i$. If there is no customer in the queue, nothing happens. She joins the infinite-sever queue where service time has an  exponential distribution with parameter  $\mu$. Then she chooses to join cluster $j$ with probability $\gamma_j$. For that,  she chooses two queues at random in cluster $j$, and the system indicates the least loaded among the two queues, ties being solved at random. She returns at this queue if it is possible. Otherwise, she remains in the infinite-server queue and thus reattempts selecting another cluster, after some time still exponentially distributed with parameter $\mu$,  then a third one if the two queues chosen in the second one are saturated, until she joins a one-server queue. All the interarrival and service times are assumed to be independent.

The key state process is then described.  Let the process $U^N(t)=(U_{i,k}^N(t),1\leq i\leq C,0\leq k\leq K_i)$ be  defined where $U_{i,k}^N(t)$ is the proportion of queues of cluster $i$ with more than or equal to $k$ customers at time $t$ (the number of queues of cluster $i$ with more than $k$ customers divided by the number $N_i$ of queues of cluster $i$). It is a time continuous Markov process irreducible on the finite state space $\displaystyle{{\cal{U}}_C=\left\{u=(u_{i,k})_{1\leq i\leq C,0\leq k\leq K_i},k\mapsto u_{i,k}\text{ decreasing},\; u_{i,k}\geq 0,\; u_{i,0}=1\right\}}$, whose $\cal{Q}$-matrix is given by, for $u\in \cal{U}_C$,
\[
\begin{cases} \cal{Q}(u,u-\frac{1}{N_i}e_{i,k}) &= \lambda_{i}N_i(u_{i,k}-u_{i,k+1})1_{k>0}\\ \cal{Q}(u,u+\frac{1}{N_i}e_{i,k+1}) &= {\mu}N(s-\sum\limits_{c=1}^{C}\alpha_c \sum\limits_{k=1}^{K_c}u_{c,k})\gamma_i(u_{i,k}^2-u_{i,k+1}^2)1_{k<K_i}.\end{cases}
\]
The first transition corresponds to the departure of a customer from a queue of cluster $i$ with $k$ customers (to queue $0$), while the second transition to the arrival of a  customer in a queue of cluster $i$ with $k$ customers (from queue $0$).

\section{Convergence Results}
\label{sec:convergence}
 As $N$ gets large, the system behaves as the solution of a dynamical system. It is given by the following proposition.

\subsection{Convergence to the dynamical system}
\begin{proposition}
If $(U^N(0))$ converges in distribution to $x$
then the Markov process $(U^N(t))$ converges in distribution to the unique solution $(u(t))$ of the following ODE, for $1\leq i\leq C$,
\begin{equation}
\left\{
\begin{aligned}
u_{i,0}(t) &=1\\
 \dot{u}_{i,k}(t) & = -\lambda_i(u_{i,k}(t)-u_{i,k+1}(t)1_{k<K_i})\\
&\quad +\frac{\mu\gamma_i}{\alpha_i}(s-\sum\limits_{c=1}^{C}\alpha_c \sum\limits_{k=1}^{K_c}u_{c,k}(t))(u_{i,k-1}^2(t)-u_{i,k}^2(t)),\;1\leq k \leq K_i
\end{aligned}
\right.
\label{sysdeuxclusters}
\end{equation}
with $u(0)=x$.
\end{proposition}
The proof is  classical and omitted.
\subsection{A unique equilibrium point for the dynamical system}\label{uniqueness}
Let us introduce the following notations. Let $1\leq i\leq C$.
\begin{align}\label{Lambda}
R_i &\stackrel {def}{=}\mu \gamma_i / \lambda_i \alpha_i\\
r_i &\stackrel{def}{=}R_i/\max\limits_{j}{R_j}\\
{\text and } \nonumber\\
 \Lambda  &\stackrel {def}{=}\displaystyle{1/\max\limits_{j}{R_j}}.
\end{align}
$R_i$ is called the utilization of a queue in cluster $i$ and $r_i$ the relative utilization of a queue in cluster $i$.

An equilibrium point of this dynamical system is thus a solution $\bar{u}$ of
\begin{equation}\label{fpe}
\left\{
\begin{array}{llll}
\smallskip
u_{i,0}=1\\
0=(u_{i,k}-u_{i,k+1}1_{k<K_i})-{\rho}r_i(u_{i,k-1}^2-u_{i,k}^2)&, 1\leq k\leq K_i .
\end{array}
\right.
\end{equation}
where 
\begin{align}\label{defrho}
\rho\stackrel{def}{=}\Lambda^{-1} (s-\sum\limits_{c=1}^{C}\alpha_c \sum\limits_{k=1}^{K_c}u_{c,k}).
\end{align}

The proof of the uniqueness of the equilibrium point, i.e. the solution of equation~\eqref{fpe} is given in   \cite{FrickerGast-velib}. The sketch of the proof is recalled here for the clarity of the exposition.

The proof is  the following.
First, let $i\in \{1,\ldots, C\}$ be fixed. There is a probabilistic interpretation for the equilibrium point $\bar{u}_i$ of the ODE~\eqref{sysdeuxclusters}. By equation~\eqref{fpe}, $\displaystyle{(\bar{u}_{i,k})_{0\leq k\leq K_i}}$ is the limiting distribution of the stationary number of customers in a queue in a system of $L$ queues with the same capacity $K_i$, where service times at each queue are independent with exponential distribution with mean $1$ and where the  arrival process at the system is a Poisson process with rate ${\rho}r_iL$, customers choosing then two queues at random among the $L$ and going to the least loaded, ties being solved at random, where $\rho$ is solution of equation~\eqref{defrho}. Inter arrival times and service times are assumed to be independent. Note that ${\rho}$ depends on all the $\displaystyle{(\bar{u}_{i,k})_{0\leq k\leq K_i}}$, $1\leq i\leq C$. 
Thus its existence and uniqueness has to be proved.
 Two main arguments are needed.
 
First, let such
 a system  of $L$ queues with a two-choice strategy be with a fixed arrival rate $\rho L$.
There exists a unique equilibrium point $\bar{v}$ for the associated dynamical system $(v(t))$, limit as $L$ gets large of the process $(V^N(t))$ where $V^N(t)$ is the vector of the proportion at time $t$ of queues with more than $k$ customers, $0\leq k\leq K$. It is given by the following lemma.

\begin{lemma}\label{2-choice-fp-existence}
 There exists a unique solution  $\bar{v}=\nu_{\rho,K}$ of 
 \begin{equation}\label{two-choice}
\left\{
\begin{array}{llll}
\smallskip
v_{0}=1\\
0=(v_{k}-v_{k+1}1_{k<K})-{\rho}(v_{k-1}^2-v_{k}^2)&, 1\leq k\leq K .
\end{array}
\right.
\end{equation} 
\end{lemma}
\begin{proof}
If $K$ is infinite, this result is known for a long time (see for example  \cite{Vvedenskaya-1}   or \cite{Mitzenmacher-1}) and have an explicit form 
\begin{align}\label{formule_magique}
\bar{v}_k=\rho^{2^{k-1}-1},\; k\geq 0.
\end{align}
In case of finite capacity, the result is given  in \cite{FrickerGast-velib}, where the proof is detailed. Nevertheless, an explicit expression is not available.
\end{proof} 

Second, the following monotonicity property is very useful.
\begin{lemma}\label{2-choice-fp-monotonicity}
$\nu_{\rho,K}$ is an increasing  function of  $\rho$.
\end{lemma}
\begin{proof}
It can be proved by a coupling argument that, if $\rho\leq \rho'$ then, for each $k$, $1\leq k\leq K$, for each $t>0$, $V^{N,\rho}_k(t)\leq_{st}V^{N,\rho'}_k(t)$.
As usual, such a proof is tedious but necessary to avoid mistakes. It has been skipped in \cite{FrickerGast-velib}. It is given here in Appendix. Taking the limit as $N$ tends to $+\infty$, it ends the proof.
\end{proof}
\begin{proposition}\label{fp_uniqueness}
There exists a unique solution $\bar{u}=\displaystyle{(\bar{u}_{i,k})_{1\leq i\leq C,\; 0\leq k\leq K_i}}$ to equation~\eqref{fpe}, given, for  $1\leq i\leq C$  and $1\leq k\leq K_i$,
$\bar{u}_{i,k}=\nu_{\rho r_i, K_i}$  where $\rho$ is the unique solution of 
\begin{align}\label{def_s}
s=\Lambda \rho+ \sum \limits_{c=1}^{C}\alpha_c \sum\limits_{k=1}^{K_c}\nu_{\rho r_c,K_c}(k).
\end{align}

\end{proposition}
\begin{proof}
 Let $\bar{u}$ be an equilibrium point of the dynamical system. Then, by definition, $\bar{u}$ is solution of ~\eqref{fpe}. By Lemma \ref{2-choice-fp-existence}, $\bar{u}_i=\nu_{\rho r_i,K_i}$ where $\rho$ is given by equation ~\eqref{defrho}.  Therefore, the existence and uniqueness of $\bar{u}$ reduces  to the existence of such a $\rho$ solution of  ~\eqref{defrho}. This equation can be rewritten as equation~\eqref{def_s}.

By Lemma ~\ref{2-choice-fp-monotonicity}, the right hand side of equation ~\eqref{def_s} is an increasing function of $\rho$, from 0 to $+\infty$. Thus there exists a unique $\rho>0$ solution of \eqref{def_s}.
\end{proof}

\subsection{Convergence of the invariant measures.}

The following proposition ensures  that  the limit as $N$ gets large of the stationary proportion of queues of cluster $i$ with more that $k$ customers is given by the equilibrium point of the ODE  i.e.  $\bar{u}_{i,k}=\nu_{\rho r_i, K_i}$.

\begin{proposition}\label{convergence_des_mesures_invariantes}
The sequence of invariant measures of $(U^N(t))$ converges as $N$ gets large to the Dirac mass at $\bar{u}$ i.e.  $U^N(\infty)$ converges to the  deterministic vector $\bar{u}$, as $N$ tends to infinity.
\end{proposition}

\subsection{Generalization to the case of incentives.}
Assume now that only a fraction $\beta_i$ of customers joining cluster $i$ choose among two queues when leaving the infinite-server queue.
Proposition~\ref{sysdeuxclusters} can be rewritten as follows.
\begin{proposition}
If $(U^N(0))$ converges in distribution to $x$
then the Markov process $(U^N(t))$ converges in distribution to the unique solution $(u(t))$ of the following ODE, for $1\leq i\leq C$,
\begin{equation}
\left\{
\begin{aligned}
u_{i,0}(t) &=1\\
 \dot{u}_{i,k}(t) & = -\lambda_i(u_{i,k}(t)-u_{i,k+1}(t)1_{k<K_i})\\
&\quad +\frac{\mu\gamma_i}{\alpha_i}\left(s-\sum\limits_{c=1}^{C}\alpha_c \sum\limits_{k=1}^{K_c}u_{c,k}(t)\right)\\
& \left(\beta_i(u_{i,k-1}^2(t)-u_{i,k}^2(t))-(1-\beta_i)(u_{i,k-1}(t)-u_{i,k}(t))\right),\;1\leq k \leq K_i
\end{aligned}
\right.
\label{sysdeuxclusters_r}
\end{equation}
with $u(0)=x$.
\end{proposition}
 Then proposition~\ref{fp_uniqueness} can also be rewritten as follows. 
\begin{proposition}\label{fp_uniqueness_r}
There exists a unique solution $\bar{u}=\displaystyle{(\bar{u}_{i,k})_{1\leq i\leq C,\; 0\leq k\leq K_i}}$ to 
\begin{equation}\label{fpe_r}
\left\{
\begin{array}{llll}
\smallskip
u_{i,0}=1\\
0=(u_{i,k}-u_{i,k+1}1_{k<K_i})-{\rho}r_i
\left(\beta_i(u_{i,k-1}^2-u_{i,k}^2)+(1-\beta_i)(u_{i,k-1}-u_{i,k})\right),\\
 1\leq k\leq K_i .
\end{array}
\right.
\end{equation}
 given, for  $1\leq i\leq C$  and $1\leq k\leq K_i$, by
$\bar{u}_{i,k}=\nu_{\rho r_i, K_i}$  where $\rho$ is the unique solution of 
\begin{align}\label{def_s_c}
s=\Lambda \rho+ \sum \limits_{c=1}^{C}\alpha_c \sum\limits_{k=1}^{K_c}\nu_{\rho r_c,K_c}(k).
\end{align}
\end{proposition}
The proof uses that Lemma~\ref{2-choice-fp-existence} remains true in the extended framework. Furthermore Proposition~\ref{convergence_des_mesures_invariantes} holds. 

To conclude, convergence results adapt to the case where $\beta\in]0,1[$.
It has a great importance in applications because in practice, only a fraction of users can follow incentives to go to a queue given by an advisor.

\section{Performance Analysis}
\label{sec:performance}
In this model, by Propositions~\ref{fp_uniqueness} and~\ref{convergence_des_mesures_invariantes}, $\rho$ is a parameter for $s$ by equation~\eqref{def_s_c} and the limiting stationary state given by the $\nu_{\rho r_i,K_i}$'s. Thus, as $N$ tends to infinity,  the limiting proportion of problematic queues in each cluster can be plotted by a parametric curve in $\rho$ as a function of the proportion $s$ of customers per queue and then the global limiting proportion of problematic queues.  The aim is to investigate  this function. It will reduce to study its behavior around the different minima of the limiting proportions of problematic queues in each cluster. In the following, the system is always studied as it gets large even if the word limiting is not mentioned.

\begin{definition}\label{perf_totale}
Let $C\geq 1$ be  fixed. Denote by $\alpha$ the vector $(\alpha_i)_{1\leq i\leq C}$ and by $K$ the vector $(K_i)_{1\leq i\leq C}$. Let 
\begin{align}\label{def_Pg}
P(C,\rho,\alpha,K)=\sum_{i=1}^C \alpha_i P(\rho r_i,K_i)
\end{align}
where  $P(\rho,K)=\nu_{\rho,K}(0)+\nu_{\rho,K}(K)$ be the limiting global proportion of problematic queues. Note that $P(\rho r_i,K_i)$ is 
the proportion of problematic queues in cluster $i$,  also denoted by $P_i(\rho)$ in the following, with a slight abuse of notation, as the dependence on $r_i$ and $K_i$ is not explicitly expressed. Note also that for $C=1$, $P(1,\rho,1,K)$ reduces to $P(\rho, K)$.
\end{definition}

Unfortunately, there is no explicit form for $\nu_{\rho ,K}$ as the capacity of the queues is finite. One can wonder whether the explicit expression of $\nu_{\rho,K}$ for $K=+\infty$ can be used as a good approximation for $K<\infty$, because in practice $K$ is equal to a few tens. In fact this approximation is far to be sufficient. Take the homogeneous case $C=1$. The behavior of the proportion of problematic queues around the minimum is described  by a very short interval $\rho<1$, $\rho$ close to $1$ (see details in  \cite[Theorem 2]{FrickerGast-velib}). For this interval, the previous approximation collapses. Theorem 2 in  \cite{FrickerGast-velib} fills this gap by describing precisely the behavior of the performance around its minimum. In the case of different clusters, the situation is more complicated.  The behavior of the curve when $\rho$ is not close to $1$ is also needed. For example for $\rho<1$, the approximation when $K=+\infty$ is used. The following section  provides such approximations. It recalls first the result of Theorem 2 in \cite{FrickerGast-velib} for the homogeneous case $C=1$, when $\rho<1$, $\rho$ close to $1$, then gives approximations in the two cases $\rho<1$ and $\rho\geq 1$.

\subsection{Preliminary results on the homogeneous case}\label{hom}

For that, let us focus on the homogeneous case $C=1$. For sake of simplicity, and  with a slight abuse of notation,  the notation  $\bar{u}$ is replaced by either $u(\rho,K)$ or $u$. The first lemma aims  to rewrite equation

\begin{lemma}
 Equation~\eqref{two-choice} is equivalent to
\begin{equation}\label{two-choice-1rec}
\left\{
\begin{array}{llll}
\smallskip
u_{0}&=1,\\
u_{k+1}&={\rho}(u_{k}^2-1)+u_1,\; 1\leq k\leq K,\\
u_{K+1}&=0.
\end{array}
\right.
\end{equation}
\end{lemma}
\begin{proof}
Indeed, by  equation~\eqref{two-choice}, it holds, for $1\leq k\leq K$,
\[u_{k}-\rho u_{k-1}^2=u_{k+1}-\rho u_{k}^2=u_1-\rho.\]
The equivalence follows.
\end{proof}
To understand the behavior of the performance, the proportion of problematic queues denoted by $P(\rho, K)$ and defined as $\nu_{\rho,K}(0)+\nu_{\rho,K}(K)$ has to be determined as long as $\sum_{k=1}^K u_k(\rho,K)$ which allows to obtain  $s$. These values are needed with the largest precision according to the value of parameter $\rho$. For $\rho<1$ and close to $1$,
\begin{lemma}\label{rho=1}
\begin{enumerate}[label=$(\roman*)$]
\item For $\rho\in[1-2^{-K/2},1]$,  $P(\rho,K)\leq4\sqrt{K}2^{-K/2}$.\\
\item   $\sum\limits_{k=1}^K{u_k}(1-2^{-K/2},K)\leq{K/2}$ and  $\sum\limits_{k=1}^K{u_k}(1,K)\geq{K-\log_2{K}-3}$.
\end{enumerate}
\end{lemma}
\begin{corollary}\label{FG_th2}
In an homogeneous system with 2-choice incentives, the proportion of problematic queues  $P(\rho,K)$ is less than $\sqrt{K}2^{-K/2}$
for all  $s\in [K/2+\lambda/\mu,K-\log_2 K-3]$.
\end{corollary}
\begin{proof}
The two previous results are proved in \cite[Theorem 2]{FrickerGast-velib}.
\end{proof}
The following lemma deals with the case $\rho<1$.
\begin{lemma}\label{rho<1}
 For $\rho<1$,
(i)  $1-\rho\leq P(\rho,K)\leq 1-\rho +2\rho^{2^{K}-1}$,\\
(ii) $S(\rho,K)-2^{K+1} \rho ^{2^{K+1}-1}\leq \sum\limits_{k=1}^K{u_k}(\rho,K)\leq S(\rho,K)$ where $S(\rho,K)\stackrel{def}{=}\sum\limits_{k=1}^K\rho^{2^k-1}$.
\end{lemma}
\begin{proof}
For $\rho\leq 1$, $u_1\leq 1$ (see details in \cite[Theorem 2]{FrickerGast-velib}). Then, using equation~\eqref{two-choice-1rec}, for $1\leq k\leq K$, $u_k\leq \rho^{2^k-1}$. 
Thus the last inequality in (ii) holds.

Let $\varepsilon=\rho-u_1$. Still using equation~\eqref{two-choice-1rec}, for $1\leq k\leq K$,
\begin{align}\label{maj}
u_{k+1}=\rho{u_k^2}-\varepsilon.
\end{align}
 Note that the proof of \cite[Theorem 2]{FrickerGast-velib} leads to $0<\varepsilon\leq{K2^{-K}}$, for each $\rho\leq 1$. If $\rho<1$,  a better bound can be obtained. Using  equation~\eqref{maj}, for $1\leq k\leq K$,
$u_{k}\leq\rho{u_{k-1}^2}$ then by induction, for $1\leq k\leq K$, $u_k\leq\rho^{2^{k-1}-1}u_1^{2^{k-1}}$. But $u_{K+1}=0$, thus
$u_{K+1}=0=\rho{u_K^2}-\varepsilon\leq\rho^{2^K-1}u_1^{2^K}-\varepsilon\leq\rho^{2^K-1}(\rho-\varepsilon)^{2^K}-\varepsilon$. Therefore
\begin{equation*}
\varepsilon\leq \rho^{2^{K+1}-1}(1-\frac{\varepsilon}{\rho})^{2^K}\leq\rho^{2^{K+1}-1}.
\end{equation*}
By a direct recurrence, already  in the proof of \cite[Theorem 2]{FrickerGast-velib}, $u_k\geq \max(\rho^{2^k-1}-(2^k-1)\varepsilon,0)$, and therefore $u_k\geq \max(\rho^{2^k-1}-2^k\varepsilon,0)$. It yields
\begin{equation*}
\sum\limits_{k=1}^K{u_k}\geq \sum\limits_{k=1}^K(\rho^{2^k-1}-2^k\varepsilon)\geq S(\rho,K)-(2^{K+1}-2)\varepsilon\geq S(\rho,K)-2^{K+1}\rho^{2^{K+1}-1}.
\end{equation*}
Furthermore, to prove  (ii),
$P(\rho,K)=1-u_1+u_K=\displaystyle{1-u_1+\sqrt{\frac{\rho-u_1}{\rho}}}=
\displaystyle{1-\rho+\varepsilon+\sqrt{\frac{\varepsilon}{\rho}}}$
where 
\begin{align}\label{lastmaj}
\displaystyle{\varepsilon +\sqrt{\varepsilon/\rho}\leq \rho^{2^{K+1}-1}+\rho^{2^K-1}\leq 2\rho^{2^K-1}}.
\end{align}
 It ends the proof.
\end{proof}

For $\rho\geq 1$, queues tend to be overloaded. It is thus interesting to introduce the number $w_k$ of queues with more than $k$ empty slots, instead of the number  of queues with more than $k$ customers. The study of $w$ leads to the following result.
\begin{lemma}\label{rho>1}
For $\rho \geq 1$,
\begin{enumerate}[label=$(\roman*)$]
\item $\displaystyle{\zeta(\rho,K)\leq \sum\limits_{k=1}^K{u_k}\leq K-\frac{1}{2\rho-1}+\eta(\rho,K)}$\\
where $\zeta(\rho,K)=\displaystyle{\sum\limits_{k=1}^{K}{x_k}}$ with $x_1=\displaystyle{\sqrt{(\rho-1)/\rho}}$, $x_{k+1}=\displaystyle{\sqrt{(x_k-1)/\rho+1}}$ 
\[\text{ and } \eta(\rho,K)=\frac{1}{(2\rho-1)(2\rho)^K}+\frac{K^{2}2^{-K}}{2\rho-1}+\frac{K2^{-K}}{(2\rho-1)^2(2\rho)^K}.\]
\item $\sqrt{1-1/\rho}\leq P(\rho,K)\leq \sqrt{1-1/\rho}+K2^{-K}(1+1/2\sqrt{\rho(\rho-1)})$.
\end{enumerate}
\end{lemma}
{\bf Remark.} First here $\zeta$ has not an explicit expression but is numerically computable. Second, the value of $\eta(\rho,K)$ is negligible for the practical values of $K$. Numerically, note first that $s$ is close to $K$ for $\rho=1.1$. The mean capacity of stations in Paris is more than $32$, with a standard deviation of 13. For $\rho=1.1$ and $K=20$, $\eta(\rho,K) \approx 0.026$ which is negligible compared to  $1$. 
 
The proof is given in Appendix.

\subsection{Results on the two-cluster case}
This section deals with the two-cluster case in order to have simple and readable results. The two clusters are numbered such that $r_1\leq r_2$. By definition, $r_2=1$. 
As investigated in Section~\ref{hom},
the interval of $\rho$ for which the proportion of problematic queues of cluster $i$ is minimal is very centered around $\rho r_i=1$. In the case of two clusters, the  trends are the following: the proportion of problematic queues in cluster $1$ is minimal on an interval corresponding to $\rho r_1$ close to $1$, i.e. for $\rho$ close to $1/r_1$. This interval does not correspond to $\rho r_2=\rho$ close to $1$, thus does not match with the region where the proportion of problematic queues in cluster $2$ is minimal, which is very centered around $\rho=1$. Therefore there are two different regions where  the proportion of problematic queues  is minimal in each cluster. Concerning the global proportion of problematic queues, there are two plateaux. The first one corresponds to $\rho=1$ and the second one to $\rho=1/r_1$. 

\subsection{The first plateau} \label{premier_palier}
A result similar to Corollary~\ref{FG_th2} holds for the first plateau corresponding to  $\rho=1$.

Recall that as given in Definition~\ref{perf_totale},  $P(\rho r_i,K_i)$ is the proportion of problematic queues in cluster $i$ given by $P(\rho r_i,K_i)=\nu_{\rho r_i,K_i}(0)+\nu_{\rho r_i,K_i}(K_i)$,  $P(C,\rho,\alpha,K)$ is the global proportion of problematic queues in the system given by equation~\eqref{def_Pg}, and $s$ is given as a function of $\rho$ by equation~\eqref{def_s_c}.

\begin{proposition}\label{prop:premier palier}
In a two-cluster system with  two-choice incentives, a first plateau  corresponds to  $\rho$ close to 1 and  $P_1(\rho)\approx 1-r_1$ while $P_2(\rho)\approx 0$ . More precisely,
\begin{align*}
\alpha_1(1-r_1)\leq P(2,\rho)\leq \alpha_1(1-r_1)+4\alpha_2\sqrt{K_2}2^{-K_2/2}+\alpha_1(r_12^{-K_2/2}+2r_1^{2^{K_1}-1})
\end{align*}
for all $s$ in 
\begin{multline*}
\left[\Lambda+\alpha_2K_2/2+\alpha_1S(r_1,K_1),\right.\\
\left. \Lambda+\alpha_2(K_2-\log_2{K_2}-3)+\alpha_1(S(r_1,K_1)-2^{K_1+1}r_1^{2^{K_1+1}-1})\right].
\end{multline*}
     \end{proposition}
\begin{proof}
Let  $\rho$ be in $[1-2^{-K_2/2},1]$. On one hand, by Lemma~\ref{rho=1}, $P(\rho,K_2)\leq4\sqrt{K_2}2^{-K_2/2}$.

On the other hand, on this interval of $\rho$, $\rho{r_1}\in[r_1(1-2^{K_2/2}),r_1]$. But
\[ P(\rho r_1,K_1)=1-\nu_{\rho r_1,K_1}(1)+\nu_{\rho r_1,K_1}(K_1).
\]
Let $\epsilon_{\rho ,K}=\rho-\nu_{\rho ,K}(1)$. By equation~\eqref{two-choice-1rec},
thus, $\nu_{\rho r_1,K_1}(K_1)=\sqrt{\epsilon_{\rho r_1 ,K_1}/(\rho r_1)}$ and then, by definition of $P(\rho,K)$,
\begin{align}\label{decomp}
P(\rho r_1,K_1)=1-r_1+r_1(1-\rho)+\epsilon_{\rho r_1 ,K_1}+\sqrt{\epsilon_{\rho r_1 ,K_1}/(\rho r_1)}.
\end{align}
As $\rho r_1<1$, by equation~\eqref{lastmaj}, the sum of the two last terms of the right-hand side of equation~\eqref{decomp} is less than $2(\rho r_1)^{2^{K_1}-1}$. Hence, 
\begin{align}\label{P_1}
1-r_1\leq P(\rho r_1,K_1)\leq 1-r_1+r_1 2^{-K_2/2}+2 r_1^{2^{K_1}-1}.
\end{align}
Furthermore, for $\rho=1-2^{K_2/2}$, $\sum_{k=1}^{K_2} \nu_{1-2^{K_2/2},K_2}(k)\leq K_2/2$ by Lemma~\ref{rho=1} and $\sum_{k=1}^{K_1} \nu_{(1-2^{K_2/2})r_1,K_1}(k)\leq S(r_1,K_1)$, by Lemma~\ref{rho<1} and the monotonicity of $S(\rho,K)$ as a function of $\rho$. Then for $\rho=1$, $\sum_{k=1}^{K_2} \nu_{1-2^{K_2/2},K_2}(k)\geq  K_2-\log_2 K_2-3$ by Lemma~\ref{rho=1} and $\sum_{k=1}^{K_1} \nu_{(1-2^{K_2/2})r_1,K_1}(k)\geq S(r_1,K_1)-2^{K_1+1}r_1^{2K_1+1}-1$, by Lemma~\ref{rho<1}. 
Then, for all  
\[s\in[\Lambda+\alpha_2K_2/2+S(r_1,K_1),\Lambda+\alpha_2(K_2-\log_2{K_2}-3)+\alpha_1(S(r_1,K_1)-2^{K_1+1}r_1^{2^{K_1+1}-1})],
\]
 $\rho\in[1-2^{-K_2/2},1]$. It ends the proof.
\end{proof}
\subsection{The second plateau} \label{second_palier}
Similarly, the following result holds. Let us denote $\rho_1=(1-2^{-K_1/2})/r_1$.
Only the case where $\rho_1>1$ will be investigated. A similar result will be proved in the other case.

\begin{proposition}
If $\rho_1>1$,  in a two-cluster system with  two-choice incentives, a second plateau  corresponds to   $\rho<1/r_1$ and close to  $1/r_1$ and  $P_2(\rho)\approx \sqrt{1-r_1}$ while $P_1(\rho)\approx 0$.
More precisely,
\begin{multline*}
\alpha_2\sqrt{1-r_1+r_12^{-K_1/2}}\leq P(2,\rho)\leq \alpha_2\sqrt{1-r_1}\\
+4\alpha_1\sqrt{K_1}2^{-K_1/2}+\alpha_2K_12^{-K_1}(1+1/(2\sqrt{\rho_1(\rho_1-1)})).
\end{multline*}
for all  $s$ in 
\begin{multline*}
\left[\Lambda\rho_1+\alpha_1 K_1/2+\alpha_2 (K_2-\frac{1}{2/r_1-1}+\eta(\rho_1,K_2)),\right.\\
\left.\Lambda/r_1+\alpha_1(K_1-\log_2 K_1-3)+\alpha_2\zeta(1/r_1,K_2)\right].
\end{multline*}
\label{prop:second palier}
\end{proposition}
\begin{proof}
Let  $\rho$ be in $[\rho_1,\frac{1}{r_1}]$. By Lemma~\ref{rho=1}, $P(\rho r_1,K_1)\leq 4\sqrt{K_1}2^{-K_1/2}$. Then, because $\rho_1>1$, by Lemma~\ref{rho>1},
 $\sqrt{1-r_1+r_12^{-K_1/2}}\leq P(\rho,K_2)\leq \sqrt{1-r_1}+K_12^{-K_1}(1+\frac{1}{2\sqrt{\rho_1(\rho_1-1)}})$.

Using the same two lemmas, taking $\rho=\rho_1$, as $\rho_1 r_1=1-2^{-K_1/2}$ is close to $1$, 
\[\sum\limits_{k=1}^{K_1}{\nu_{1-2^{-K_1/2},K_1}}(k) \leq \sum\limits_{k=1}^{K_1}{\nu_{1,K_1}}(k) \leq  K_1/2,
\]
 and 
\[
\displaystyle{\sum\limits_{k=1}^{K_2}{\nu_{\rho_1,K_2}}(k)\leq K_2-\frac{1}{2/r_1-1}+\eta(\rho_1,K_2))}
\]
and taking  $\rho=1/r_1$, $\sum\limits_{k=1}^{K_1}{\nu_{1,K_1}}\geq K_1-\log_2 K_1-3$ and $\sum\limits_{k=1}^{K_2} \nu_{1/r_1,K_2}(k)\geq \zeta(1/r_1,K_2)$.
Putting the last two arguments together ends the proof.
\end{proof}
\subsection{Plotting the performance} 
First, the proportion of problematic queues $P_i(\rho)$ in each cluster $i$ is plotted as a function of $s$. The value of the different parameters are indicated in Figure~\ref{fig: deux clusters 1}.

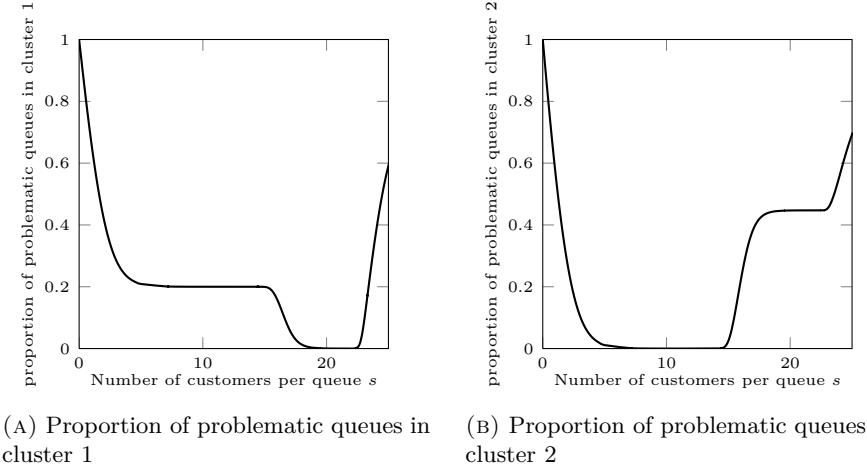
\begin{figure} [h]
   \centering      
        \begin{subfigure}[b]{0.45\textwidth}
                \centering
                \begin{tikzpicture}
                 \begin{axis}[font=\tiny,
                                       label style={font=\tiny},
      			           xlabel={Number of customers per queue $s$},
			           xlabel style={yshift=+9pt},		
                                       ylabel= proportion of problematic queues in cluster 1,
                                       ylabel style={yshift=-16pt},
                                       xmin=0, xmax=25, ymin=0., ymax=1,
                                      width=\linewidth,height=\linewidth]
                \addplot[thick, smooth] table[x=x, y=y] {donnees/Prob1.dat};
                \end{axis}
      \end{tikzpicture}
       \hspace{0.05\linewidth}
       \caption{ Proportion of problematic queues  in cluster 1}
       \label{fig:Prob1}
        \end{subfigure}
   \quad
  \begin{subfigure}[b]{0.45\textwidth}
                \centering
                \begin{tikzpicture}
                 \begin{axis}[font=\tiny,
                                       label style={font=\tiny},
      			           xlabel={Number of customers per queue $s$},
			           xlabel style={yshift=+9pt},		
                                       ylabel= proportion of problematic queues in cluster 2,
                                       ylabel style={yshift=-16pt},
                                       xmin=0, xmax=25, ymin=0., ymax=1,
                                      width=\linewidth,height=\linewidth]
                \addplot[thick, smooth] table[x=x, y=y] {donnees/Prob2.dat};
                \end{axis}
      \end{tikzpicture}
       \hspace{0.05\linewidth}
       \caption{Proportion of problematic queues in cluster 2}
       \label{fig:Prob2}
        \end{subfigure}

        \caption{Proportion of problematic queues in cluster as a fonction of the mean number of customers per queue $s$, with $\Lambda=1,3$, $K_1=20$, $K_2=25$, $r_1=0,8$, $r_2=1$, $\alpha_1=0,4$, $\alpha_2=1-\alpha_1=0,6$.}\label{fig: deux clusters 1}
\end{figure}
Both regions analysed in Sections~\ref{premier_palier} and ~\ref{second_palier} are observed for both clusters:
\begin{itemize}
\item A region where the proportion of problematic queues is very low, as for the 2-choice regulated homogeneous model. It corresponds to the interval $[18,22]$ of values of $s$ for cluster $1$ (Figure~\ref{fig:Prob1}) and  $[7,14]$ for cluster 2 (Figure~\ref{fig:Prob2}).
\item  A region where the proportion of problematic queues is quasi constant, which matches the  region where it is minimal for the other cluster. 
 For cluster 1, it is interval $[7,14]$ and on this interval $P_1(\rho)\approx 1-r_1=0,2$. For cluster 2, on $P_2(\rho)\approx \sqrt{1-r_1}\approx 0,45$ on $[18,22]$.
\end{itemize}

Then Figure~\ref{fig:Prob} plots the global proportion of problematic queues for the model with the same parameters. Note the presence of the two plateaux:
\begin{itemize}
\item the first one for $\rho<1$, $\rho$ close to  $1$, due to  cluster1-queues. By equation~\eqref{perf_totale}, 
\begin{equation*}
P(\rho)\approx \alpha_1P_1(\rho,K_1)\approx \alpha_1 (1-r_1)=0,08.
\end{equation*}
It corresponds to the interval $[7,14]$ of $s$ . By blue dashed lines, the interval of $s$ $[9.42,12.34]$ determined in  Proposition~\ref{prop:premier palier} is drawn. Recall that on this interval, explicit bounds on performance are given.
\item a second one for $\rho$ close to $1/r_1$, due to   cluster2-queues. On this interval,
\begin{equation*}
P(\rho)\approx \alpha_2P_2(\rho,K_2) \approx \alpha_2\sqrt{1-r_1}\approx 0,27.
\end{equation*} 
This plateau corresponds to  $s$ in $[18,22]$. Green dashed lines draw bounds obtained in  Proposition~\ref{prop:second palier}, i.e. $s$ in $[19.9,21.1]$.
\end{itemize}

\subsection{Influence of the model parameters}
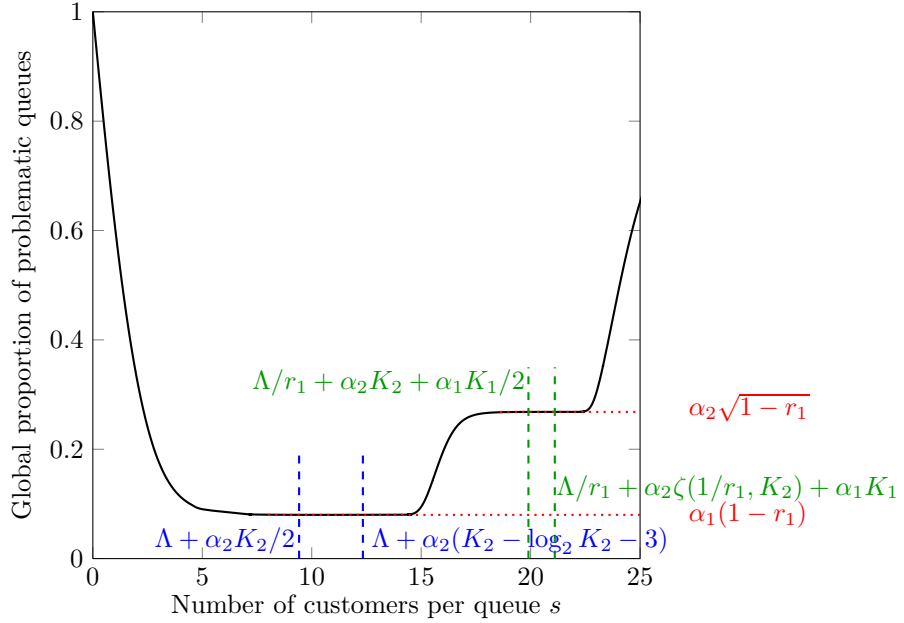
\begin{figure}[h]
    \centering
        \begin{subfigure}[b]{0.7\textwidth}
                \centering
                \begin{tikzpicture}
                 \begin{axis}[xlabel={Number of customers per queue $s$},
			           xlabel style={yshift=+5pt},		
                                       ylabel=Global proportion of problematic queues,
                                       ylabel style={yshift=-9pt},
                                       xmin=0, xmax=25, ymin=0., ymax=1,
                                      width=\linewidth,height=\linewidth]               
                \addplot[thick, smooth] table[x=x, y=y] {donnees/Prob.dat};
                \addplot[red,dotted,thick,smooth] table[x=x,y=y] {donnees/premierpalier.dat};
                \addplot[red,dotted,thick,smooth] table[x=x,y=y] {donnees/secondpalier.dat};
                \addplot[blue,dashed,thick,smooth] table[x=x,y=y] {donnees/borne1.dat};
                \addplot[blue,dashed,thick,smooth] table[x=x,y=y] {donnees/borne2.dat};
                \addplot[green!60!black,dashed,thick,smooth] table[x=x,y=y] {donnees/borne3.dat};
                \addplot[green!60!black,dashed,thick,smooth] table[x=x,y=y] {donnees/borne4.dat};
                \end{axis}
      \end{tikzpicture}
      \put(10,40){\textcolor{red}{$\alpha_1(1-r_1)$}}
      \put(10,80){\textcolor{red}{$\alpha_2\sqrt{1-r_1}$}}
       \put(-192,30){\textcolor{blue}{$\Lambda+\alpha_2 K_2/2$}}
      \put(-110,30){\textcolor{blue}{$\Lambda+\alpha_2 (K_2-\log_2 K_2-3)$}}
      \put(-155,90){\textcolor{green!60!black}{$\Lambda/r_1+\alpha_2 K_2+\alpha_1 K_1/2$}}
      \put(-40,50){\textcolor{green!60!black}{$\Lambda/r_1+\alpha_2 \zeta(1/r_1,K_2)+\alpha_1 K_1$}} 
       \hspace{0.05\linewidth}
       \end{subfigure}
        \caption{ Global proportion of problematic queues as a fonction of the mean number of customers per queue $s$.}
       \label{fig:Prob0} 
        \end{figure}

\subsubsection{ Influence of $\Lambda$.}
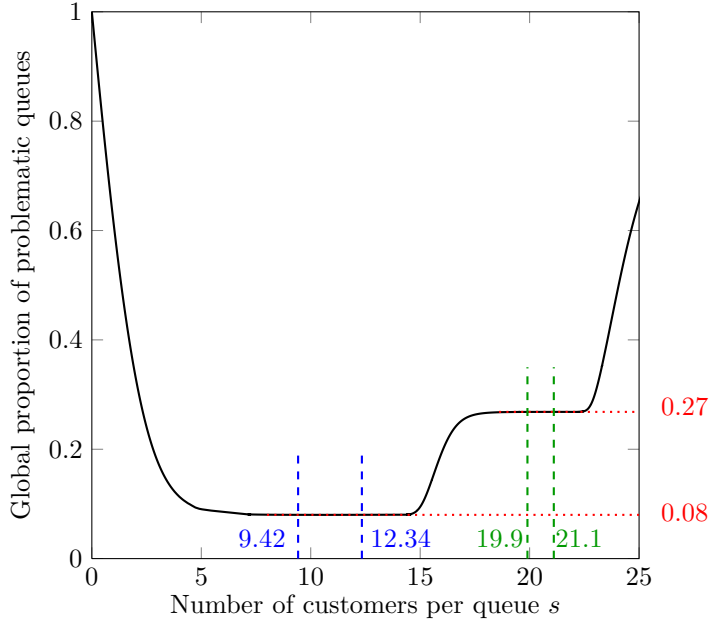
\begin{figure}[h]
    \centering
        \begin{subfigure}[b]{0.7\textwidth}
                \centering
                \begin{tikzpicture}
                 \begin{axis}[xlabel={Number of customers per queue $s$},
			           xlabel style={yshift=+5pt},		
                                       ylabel=Global proportion of problematic queues,
                                       ylabel style={yshift=-9pt},
                                       xmin=0, xmax=25, ymin=0., ymax=1,
                                      width=\linewidth,height=\linewidth]               
                \addplot[thick, smooth] table[x=x, y=y] {donnees/Prob.dat};
                \addplot[red,dotted,thick,smooth] table[x=x,y=y] {donnees/premierpalier.dat};
                \addplot[red,dotted,thick,smooth] table[x=x,y=y] {donnees/secondpalier.dat};
                \addplot[blue,dashed,thick,smooth] table[x=x,y=y] {donnees/borne1.dat};
                \addplot[blue,dashed,thick,smooth] table[x=x,y=y] {donnees/borne2.dat};
                \addplot[green!60!black,dashed,thick,smooth] table[x=x,y=y] {donnees/borne3.dat};
                \addplot[green!60!black,dashed,thick,smooth] table[x=x,y=y] {donnees/borne4.dat};
                \end{axis}
      \end{tikzpicture}
    \put(0,40){\textcolor{red}{$0.08$}}
      \put(0,80){\textcolor{red}{$0.27$}}
       \put(-160,30){\textcolor{blue}{$9.42$}}
      \put(-110,30){\textcolor{blue}{$12.34$}}
      \put(-70,30){\textcolor{green!60!black}{$19.9$}}
      \put(-40,30){\textcolor{green!60!black}{$21.1$}} 
           \hspace{0.05\linewidth}
       \end{subfigure}
        \caption{ Global proportion of problematic queues as a fonction of the mean number of customers per queue $s$, with $\Lambda=1,3$, $K_1=20$, $K_2=25$, $r_1=0,8$, $r_2=1$, $\alpha_1=0,4$, $\alpha_2=1-\alpha_1=0,6$.}
       \label{fig:Prob} 
        \end{figure} 
\begin{figure}[h]
    \centering
        \begin{subfigure}[b]{0.7\textwidth}
                \centering
                \begin{tikzpicture}
                 \begin{axis}[font=\tiny,
                                       label style={font=\tiny},
      			           xlabel={Number of customers per queue $s$},
			           xlabel style={yshift=+6pt},		
                                       ylabel=Global proportion of problematic queues,
                                       ylabel style={yshift=-6pt},
                                       xmin=0, xmax=30, ymin=0., ymax=1,
                                      width=\linewidth,height=\linewidth]               
                \addplot[red,thick, smooth] table[x=x, y=y] {donnees/Prob.dat};
                \addlegendentry{\footnotesize $\Lambda=1,3$};
                \addplot[blue,dashed, thick, smooth] table[x=x, y=y] {donnees/Probgrandlmu.dat};
                \addlegendentry{\footnotesize $\Lambda=5$};
               \end{axis}
      \end{tikzpicture}
       \hspace{0.05\linewidth}
       \end{subfigure}
        \caption{Influence of $\Lambda$ on system performance, with  $K_1=20$, $K_2=24$, $r_1=0,8$, $r_2=1$, $\alpha_1=0,4$and  $\alpha_2=1-\alpha_1=0,6$.}
          \label{fig:influence Lambda}
        \end{figure}
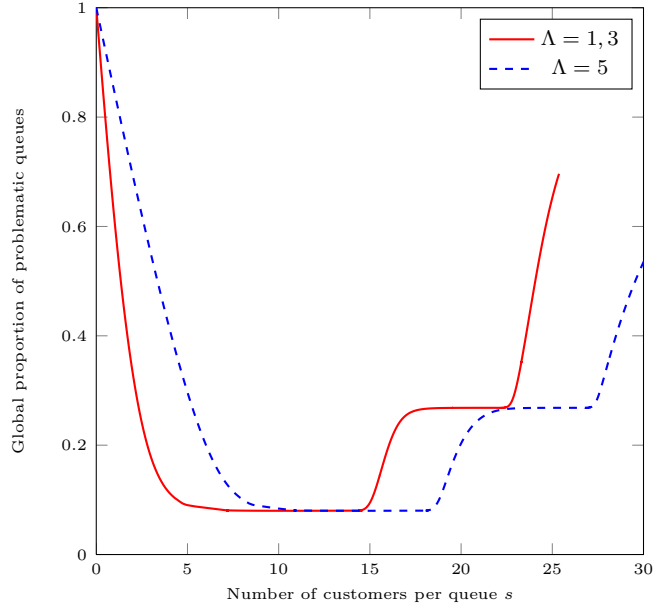       

\bigskip
\bigskip
Eventually,  Figure~\ref{fig:influence Lambda} represents the influence of $\Lambda$ on the model. By  Definition~\ref{Lambda}, $\Lambda$ is directly related to the  load characteristics of the different clusters  and generalizes  $\lambda/\mu$ when $C=1$. By equation~\eqref{def_s_c},  $\Lambda$ appears in the expression of $s$ only by the term $\Lambda \rho$, increasing  $\Lambda$ will shift the curve to the right just slightly modifying  the different interval shape  or  size. 

\subsubsection{Influence of $K$.} The different cacacities do not change the performance for the plateaux. They change the intervals where these performances are achieved. The intervals are shifted to the right where the $K_i$ are larger.

\subsubsection{Influence of the utilizations.}  A large imbalance between queues increases the probability of problematic queues. The second plateau is also shifted to the right.

\subsubsection{Influence of the cluster sizes.} A small cluster with the lowest utilisation, $\alpha_1 small$, gives a lower proportion of problematic queues. Moreover, while $\alpha_2$ is large, the plateaux are shifted to the right.

To conclude, the worse configuration is a small cluster of  queues with high utilization with respect to a large cluster of queues with low utilization.

\section{Dimensioning the System}
\label{sec:sizing}
Dimensioning  a bike-sharing system means  choosing both the capacities and the value of $s$ in order to have a minimum proportion of problematic queues. 
The fluctuations in time of the parameters is one of the problems to manage the system.

Fluctuations of parameters like arrival rates or  probabilities to return in queues due for example to flows between housing and working areas,  peak activities, etc. just affect quantity $\Lambda$. The clusters are assumed to be fixed during the time, thus the vectors $\alpha$ and $K$ are fixed. Therefore, the  problem reduces to obtain a system able to work for the different values of utilization, more precisely from a very low $\Lambda_{min}$, close to $0$ to a maximum $\Lambda_{max}$ value of $\Lambda$.  

{\bf One-cluster case.} 

One has to design   a system able to work efficiently for different arrival rates. We consider a two-choice regulated homogeneous system where the capacity $K$ of the queues and also the proportion $s$ of customers per queue has to be determined. One has to manage a system working from $\lambda_{min}$ to $\lambda_{max}$. The mean trip time does not vary over time. In Corollary~\ref{FG_th2}, an interval for $s$ $[K/2+\lambda/\mu,K-\log_2 K-3]$ is determined where the performance is very good. $K$ must be chosen such that both intervals for $\lambda=\lambda_{min}$ and  $\lambda=\lambda_{max}$ have a non empty intersection, that is written
\begin{equation*}
K/2+\lambda_{max}/\mu \leq \lambda_{min}\mu+K-\log_2{K}-3.
\end{equation*}
It is equivalent to
\begin{equation*}
K/2-\log_2{K}\geq (\lambda_{max}-\lambda_{min})/\mu+3.
\end{equation*}
Function $\psi:x\mapsto x/2-\log_2{x}$ is strictly increasing  on $[2/\ln{2},+\infty[$ and  positive for  $x\geq 4$. It admits an inverse function defined on $]0,+\infty[$ and it is just obtained when  $K\geq \psi^{-1}((\lambda_{max}-\lambda_{min})/\mu+3)$. For example, when $\lambda_{max}-\lambda_{min}/\mu=10$, take just $K\geq 37$. Then $s$ must be taken between $K/2+\lambda_{max}/\mu$ and $\lambda_{min}\mu+K-\log_2{K}-3$.

Note that in practice, bounds are not so tight even if they provide an interval for $s$ where the performance is very weak. In fact this interval is very close to 
$[K/2+\lambda/\mu,K+\lambda/\mu]$. A lower minimal value of $K$ can be expected. In the previous numerical  example, this minimal value is not far from $20$.

 It is commonly admitted that in Paris $s$ is fixed to $0.5K$ while in Lyon to $0.7K$, $K$ being a mean value of the queue capacities. A data analysis would be necessary to discuss the relevancy of these choices. Of course, a homogeneous model is a very rough picture of the system.

{\bf Two-cluster case.} 
The same kind of dimensioning could be attempted in this case. Because if the two plateaux, the idea is to manage the system in order to obtain the lowest proportion of problematic queues. There are two cases according which of the two plateaux give the best performance. the study is limited to one of these two cases, the other one being more intricate.

If $\alpha_1(1-r_1)<\alpha_2\sqrt{1-r_1}$ i.e. $\alpha_1\sqrt{1-r_1}/\alpha_2<1$ then one has to choose $K$ such that the first plateaux for $\Lambda=\Lambda_{min}$ and $\Lambda=\Lambda_{max}$
have a non empty intersection, that is written
\begin{equation*}
\alpha_2 K_2/2+\Lambda_{max}\leq \Lambda_{min}\mu+\alpha_2(K_2-\log_2{K_2}-3).
\end{equation*}
It is equivalent to
\begin{equation*}
K_2/2-\log_2{K_2}\geq (\Lambda_{max}-\Lambda_{min})/\alpha_2+3
\end{equation*}
then to 
\begin{equation*}
K \geq \psi^{-1}((\Lambda_{max}-\Lambda_{min})/\alpha_2+3).
\end{equation*}
Moreover, for such a $K_2$, $s$ must be taken between $\alpha_2 K_2/2+\Lambda_{max}/\mu$ and $\Lambda_{min}\mu+\alpha_2(K_2-\log_2{K_2}-3)$.
For example, when $\Lambda_{max}-\Lambda_{min}=10$, $\alpha_2=1/2$, take just $K\geq 57$ and for $s$ in the previous interval, the proportion of problematic queues is minimal equal to $\alpha_1(1-r_1)$.

\section{Conclusion}
\label{sec:conclusion}
 Via the analysis of the large closed network with two clusters and a two-choice strategy, the behavior of this model is quite well understood. The results underscore that, due of the heterogeneity, the system does not perform well in the sense that there there always a cluster with a non negligible proportion of problematic queues. It cannot disappear by sizing the system. Capacities and fleet size can be ajusted to obtain the lowest proportion of problematic queues, but this value is a function of the cluster size and imbalance of the stations.

In the dimensioning problem, knowing precisely the intervals where the proportion of problematic queues is low is crucial. Bounds derived for that purpose do not seem to be tigh and need to be improved. Moreover, this analysis should be extended to any number of clusters.

In applications, due to  incentives, only a fraction of users follow the rule. Though the  convergence results hold, the behavior remains to understand. This is a challenging problem for future work.

\section{Appendix}
\label{sec:appendix}
{Proof of Lemma ~\ref{2-choice-fp-monotonicity}}
\begin{proposition}
If $\rho\leq \rho'$ then, for each $k$, $1\leq k\leq K$, for each $t>0$, 
\[
V^{N,\rho}_k(t)\leq_{st}V^{N,\rho'}_k(t).
\]
\end{proposition}
\begin{proof}
 Note that there is no explicit expression for the $\nu_{\rho,K}(k)$'s for $K<+\infty$. If $K=+\infty$, using this expression~\eqref{formule_magique}, it is clear that, for each $k$, $\nu_{\rho,K}(k)$ is an increasing function of $\rho$.  It leads to prove this strongest result.

This result is proved by coupling. Let $\rho$ and $\tilde{\rho}$ be such that  $\rho<\tilde{\rho}$. Blue customers arrive according to a Poisson process with parameter $\rho N$. Independently, red customers arrive according to a Poisson process with parameter $(\tilde{\rho}-\rho )N$. Take two join-the-shortset-queue-among-two systems with $N$ queues described  as previously. In system 1, the arriving process is the process with blue customers. In system 2, the arriving process is the superposition of the processes with blue and red customers, which is Poisson  with parameter $\tilde{\rho}$. In both systems, blue customers have the same arriving and service times, and choose the same two queues, ties being solved with the same Bernoulli random variables.

Let us define the two following operations:\\
- {\em exchange} a red and a blue customer means that they change both color and residual service times.\\
- {\em  repaint} a red in blue occurs at the arrival of a blue customer. This latest one is lost, but the red one (already queuing) takes his color (blue) and service time.

We construct a coupling {\em repainting} sometimes a red customer in blue or {\em exchanging} a red and a blue customer, such that, at each time, at each queue, the blue customers are the {\em same} in both systems. {\em The same} means there is the same number of customers, in the same order, with the same residual service times, and at the beginning of the queue (red are always behind).  Let this assertion at time $t$ be called $\cal{A}(t)$. It implies that for each $t>0$, $L(t)\leq \tilde{L}(t)$.

 For that we prove that, if $\cal{A}$ is true just before an arrival or a departure of a customer, then it is true after this time. It is obvious at a  departure time of a customer or at an arrival time of a red customer.

At an arrival time of a blue customer, say time $t$, let the two choosen queues be $i$ and $j$. Assume that  $\tilde{L}_i(t_-)\leq \tilde{L}_j(t_-)$. Recall $\cal{A}(t-)$ is true. Prove $\cal{A}(t)$ distinguishing different cases.

\begin{itemize}
\item[(a)] If $\tilde{L}_i(t_-)< \tilde{L}_j(t_-)$ and $L_i(t_-)< L_j(t_-)$, then the blue customer is accepted in queue $i$ in both systems. So $\cal{A}(t)$
holds.

\item[(b)] If $\tilde{L}_i(t_-)< \tilde{L}_j(t_-)$ and $L_i(t_-)= L_j(t_-)$, 
 then the blue customer is accepted in queue $i$  for system 2. In system 1, there is a tie. If it is solved with the blue customer in queue $i$,  $\cal{A}(t)$
holds. Otherwise it is solved with the blue customer in queue $j$. But there is at least one red customer in queue $j$, so repaint the first one in the queue in blue and the arriving blue customer (in queue $i$ in system 2) in red. It means also that the service time of the arriving blue customer is exchanged with the residual service time of the red one. Notice that this residual service time has also an exponential distribution with parameter $1$. Thus  $\cal{A}(t)$
holds. 
\item[(c)] If $\tilde{L}_i(t_-)= \tilde{L}_j(t_-)<K$ and $L_i(t_-)\not = L_j(t_-)$, for example $L_i(t_-)< L_j(t_-)$. Thus in system 1, the arriving blue customer goes to queue $i$. There is a choice in system 2 and do as in case (b).
\item[(d)] If $\tilde{L}_i(t_-)= \tilde{L}_j(t_-)=K$ and $L_i(t_-) = L_j(t_-)=K$, the blue customer is rejected in both systems. $\cal{A}(t)$
holds. 
\item[(e)] If $\tilde{L}_i(t_-)= \tilde{L}_j(t_-)=K$ and $L_i(t_-)$ or  $L_j(t_-)<K$. The arriving blue customer is accepted in system 1, for example in queue $i$. But the arriving blue customer is rejected in system 2. Nevertheless, in this case,  there is one red customer in queue $i$ thus the first one  is painted in blue, his  remaining service time becomes the same as the service time of the blue customer accepted in system 1. It is clear that his total service time has an exponential  distribution with parameter 1.

\item[(f)] If $\tilde{L}_i(t_-)< \tilde{L}_j(t_-)$ and $L_i(t_-)>L_j(t_-)$, in system 1, the  blue customer is accepted at queue $i$, in system 2 at queue $j$. But in system 2, there is at least one red customer in $j$, thus {\em exchange} the first red one with the arriving blue one as in case (b). Thus $\cal{A}(t)$
holds.
\end{itemize}
 Moreover, it remains to check that process $(\tilde{L}(t))$ considered in this coupling is indeed a join-the-shortset-queue-among-two system with parameter $\tilde{\rho}N$. For that, it is sufficient to check that the random variable $X=\sigma 1_{\sigma<\tau}+(\tau+\sigma) 1_{\sigma>\tau}$
where $\sigma$, $\sigma'$ and $\tau$ are i.i.d. random variables,  with exponential  distribution with parameters $1$, $1$ and $\rho$,
has an exponential  distribution with parameters $1$. It is straightforward, replacing first $\tau$ by $t$.
\end{proof}

Letting $N$ tending to infinity, then $t$, it gives that, for each $k\geq 1$, $u_k$ is an increasing function of $\rho$.

\subsection{Proof of Lemma~\ref{rho>1}}

It yields from definition that, for $0\leq i\leq K$,  $w_i=1-u_{K-i+1}$. Plugging in equation ~\eqref{two-choice-1rec}, for $0\leq i\leq K$,  $w_{i+1}=1-\sqrt{1-(w_i-w_K)/\rho}$. In particular, $w_0=1-u_{K+1}=1$ and  $w_K=1-u_1\leq K2^{-K}$. Indeed,  using that for $\rho=1$, $1-u_1=\epsilon\leq K2^{-K}$ (see the proof of \cite[Theorem 2]{FrickerGast-velib} for details), and applying Lemma ~\ref{2-choice-fp-monotonicity}, which implies that $u$ is an increasing function of $\rho$, component by component, the following equation holds,
\begin{align}\label{daim}
0\leq \delta\leq K 2^{-K}
\end{align}
But, by induction then simple algebra, for $0\leq i\leq K$,
\begin{align*}
w_{i+1}& \geq  1-(1-\frac{w_i-w_K}{2\rho})\\
& \geq  \frac{w_i-w_K}{2\rho}\\
& \geq  \frac{w_0}{(2\rho)^{i+1}}-\frac{w_K}{2\rho} \sum\limits_{k=0}^i{(\frac{1}{2\rho})^k}\\
& \geq  \frac{1}{(2\rho)^{i+1}}-\frac{w_K}{2\rho} \frac{1-(1/(2\rho))^{i+1}}{1-1/(2\rho)}.
\end{align*}
For $\sum\limits_{k=1}^K{u_k}(\rho,K)=K-\sum\limits_{k=1}^K{w_{K-k+1}}(\rho,K)=K-\sum\limits_{i=1}^K{w_i}$,
     summing the previous inequalities for $i$ from $1$ to $K$, it holds that
\begin{align*}
\sum\limits_{k=1}^K{u_k}(\rho,K)&\leq K-\frac{\frac{1}{2\rho}-(\frac{1}{2\rho})^{K+1}}{1-\frac{1}{2\rho}}+\frac{w_K}{2\rho-1}(K-\frac{\frac{1}{2\rho}-(\frac{1}{2\rho})^{K+1}}{1-\frac{1}{2\rho}})\\
&\leq K-\frac{1}{2\rho-1}+\eta(\rho,K).
\end{align*}

As  $u_1-1\leq 0$, using again equation ~\eqref{two-choice-1rec},
\begin{equation*}
u_K=\displaystyle{\sqrt{\frac{\rho-u_1}{\rho}}=\sqrt{\frac{\rho-1}{\rho}+\frac{u_1-1}{\rho}}\geq \sqrt{\frac{\rho-1}{\rho}}}.
\end{equation*}
Besides, as obtaining equation ~\eqref{two-choice-1rec}, from  equation ~\eqref{two-choice}, for $1\leq k\leq K$, $u_{k+1}-\rho u_k^2=u_{K+1}-\rho u_K^2=-\rho u_K^2$. Thus, 
\begin{equation*}
u_k=\displaystyle{\sqrt{u_{k+1}/\rho+u_K^2}\geq \sqrt{(u_{k+1}/\rho-1)+1}}.
\end{equation*}

By induction, it is then easy to prove that for  $0\leq k\leq K$, $u_k\geq x_{K-k+1}$. Summing from $1$ to $K$ yields the first inequality of (i).
Let us prove (ii). By definition,
\begin{align*}
P(\rho,K)& =1-u_1+u_K\\
& =1-u_1+\sqrt{\frac{\rho-u_1}{\rho}}\\
& =\delta+\sqrt{1-\frac{1}{\rho}+\frac{\delta}{\rho}}
\end{align*}
Plugging equation \eqref{daim} in it leads to (ii). It  ends the proof.

\end{document}